\documentclass{article}
\usepackage{arxiv}
\usepackage{graphicx}%
\usepackage{multirow}%
\usepackage{amsmath,amssymb,amsfonts}%
\usepackage{amsthm}%
\usepackage{mathrsfs}%
\usepackage[title]{appendix}%
\usepackage{xcolor}%
\usepackage{textcomp}%
\usepackage{manyfoot}%
\usepackage{booktabs}%
\usepackage{algorithm}%
\usepackage{algorithmicx}%
\usepackage{algpseudocode}%
\usepackage{listings}%
\usepackage{tikz,physics,bm}

\usetikzlibrary{positioning}
\usetikzlibrary{decorations.text}
\usetikzlibrary{decorations.pathmorphing}
\usepackage[short,c2]{optidef}

\usepackage{empheq}
\usepackage{cases}
\newtheorem{theorem}{Theorem}
\newtheorem{proposition}[theorem]{Proposition}%
\newtheorem{lemma}[theorem]{Lemma}%

\newtheorem{example}{Example}%
\newtheorem{remark}{Remark}%

\title{Distribution of Age of Information \\ in the Erlang Loss System}

\author{
    Nail~Akar\\
	Electrical and Electronics Engineering Dept.\\
	Bilkent University \\
	Ankara, Turkey \\
	\texttt{akar@ee.bilkent.edu.tr} \\
    \And
    Sennur~Ulukus\\
    Electrical and Computer Engineering Dept. \\
    University of Maryland \\
    College Park, MD, USA \\
    \texttt{ulukus@umd.edu} 
	 }

\begin{document}
\maketitle

\begin{abstract} In this paper, we study the exact distributions of the age of information (AoI) and peak AoI (PAoI) in a bufferless setting in which time-stamped updates, or processed tasks, generated by one or several information sources according to a Poisson process for each source, are submitted to a shared pool of $c$ homogeneous servers with exponentially distributed service times, i.e., the so-called M/M/c/c or the Erlang loss system.
We consider three update management policies which come into play when a new update arrives to find all the servers busy: the arriving update is blocked (non-preemptive, NP) as in the Erlang loss system, or it preempts a randomly chosen update of its own source in service (preempt at random, PR), or the stalest such update (preempt the stalest, PS). All three policies preserve the same birth-death structure of the server occupancy process associated with the underlying Erlang loss system. However, their AoI and PAoI distributions can be very different.
The approach we take is the absorbing Markov chain (AMC) method, in which a single AoI cycle, rather than the entire sample path of the system, is modeled by an absorbing Markov chain. Via the AMC method, we derive the exact distributions of AoI and PAoI in matrix-exponential form.
The method extends to multiple sources sharing the same pool of servers, with the AoI of a given source being affected by the remaining sources only through their aggregate update rate.
Numerical examples illustrate the implications of our findings, including distribution-based server provisioning under age violation constraints.\end{abstract}




\maketitle

\section{Introduction}
\label{sec:intro}
In a growing range of applications, a monitor makes real-time inferences or control decisions based on the outcomes of {\em processed tasks} which are time-stamped units of work, such as sensor readings to be fused, camera frames to be analyzed by an inference engine, or measurements to be validated.
These tasks are submitted to a pool of processing resources, with the results used by the monitor upon completion of processing \cite{li_etal_iotj21,xu_etal_infocom22,abbas_survey23}. 
In our setting, what matters to the monitor in these applications is not necessarily how long an individual task spends in the system, but rather how {\em fresh} the most recently delivered result is. Hence, a computed result is informative as long as no result carrying a more recent time stamp has been delivered. The same concern with freshness is also associated with {\em status update systems} in which information sources send time-stamped updates to remote monitors over communication links or networks introducing random delays and losses \cite{elmagid_commag19,RoyYates__AgeOfInfo_Survey}. In both settings, conventional network performance metrics, such as delay, loss, and throughput, fail to capture timeliness as perceived by the monitor, a shortcoming that has led to the emergence of several novel information freshness metrics \cite{RoyYates__AgeOfInfo_Survey}. In this paper, as the information freshness criterion, we employ the age of information (AoI), first defined in \cite{kaul_etal_infocom12}, which is a random process whose value at time $t$ is the elapsed time, or age, since the generation time of the freshest update received by the monitor up to time $t$. Particularly, the AoI process is composed of AoI cycles. Within a cycle, the process increases with unit slope, and at the end of the cycle, i.e., upon the reception of a fresher update, the process experiences a downward jump. The peak AoI (PAoI) process is constructed by sampling the peak value attained in each cycle. One particular feature of the AoI concept is its unawareness of the internal dynamics of the information source, making it a convenient tool for the design of update systems without requiring a statistical characterization of the source. There are also freshness metrics that actually account for source dynamics including binary freshness \cite{akar_ulukus_tcom24,Sennur__BF_InfoFreshnessInCacheUpdating}, age of incorrect information \cite{maatouk2020,kam2020age}, and age of synchronization \cite{zhong_etal_isit18,AoS_2__SchedulingToMinimizeAoSInWireless}; see \cite{RoyYates__AgeOfInfo_Survey,kosta_etal_survey,abbas_survey23} for surveys of AoI modeling, analysis, and optimization.

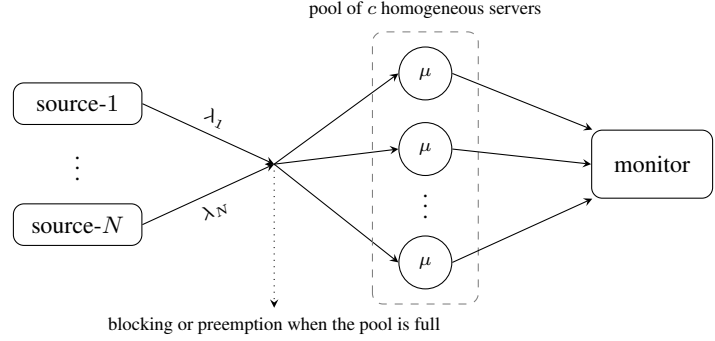
\begin{figure}[tb]
    \centering
    \begin{tikzpicture}[>=stealth,font=\small]
    \node[draw, rounded corners, minimum width=1.7cm, minimum height=0.55cm] (s1) at (0,1.8) {source-$1$};
    \node at (0,1.05) {$\vdots$};
    \node[draw, rounded corners, minimum width=1.7cm, minimum height=0.55cm] (sN) at (0,0.2) {source-$N$};
    \coordinate (j) at (2.6,1.0);
    \draw[->] (s1.east) -- node[above,sloped] {\scriptsize $\lambda_1$} (j);
    \draw[->] (sN.east) -- node[below,sloped] {\scriptsize $\lambda_N$} (j);
    \node[draw, circle, minimum size=0.7cm] (c1) at (4.6,2.2) {\scriptsize $\mu$};
    \node[draw, circle, minimum size=0.7cm] (c2) at (4.6,1.2) {\scriptsize $\mu$};
    \node at (4.6,0.55) {$\vdots$};
    \node[draw, circle, minimum size=0.7cm] (cc) at (4.6,-0.3) {\scriptsize $\mu$};
    \draw[->] (j) -- (c1.west);
    \draw[->] (j) -- (c2.west);
    \draw[->] (j) -- (cc.west);
    \draw[dashed, rounded corners, gray] (3.9,-0.85) rectangle (5.3,2.75);
    \node at (4.6,3.05) {\scriptsize pool of $c$ homogeneous servers};
    \node[draw, rounded corners, minimum width=1.6cm, minimum height=0.9cm] (m) at (7.6,1.0) {monitor};
    \draw[->] (c1.east) -- (m.150);
    \draw[->] (c2.east) -- (m.west);
    \draw[->] (cc.east) -- (m.210);
    \draw[->, dotted] (j) -- (2.6,-0.9) node[below] {\scriptsize blocking or preemption when the pool is full};
    \end{tikzpicture}
    \caption{The multi-source bufferless update system studied in this paper. Source-$n$ generates time-stamped tasks according to a Poisson process with intensity $\lambda_n$. Each admitted update immediately starts service at one of the $c$ homogeneous servers with exponentially distributed service times, and the result is delivered to a monitor as long as it is the freshest result. An update arriving to a full pool is blocked (NP), or preempts an update of its own source in service (PR, PS).}
    \label{fig:system}
\end{figure}

In this paper, we study the AoI and PAoI processes in the multi-server bufferless system depicted in Fig.~\ref{fig:system} in which updates arrive according to a Poisson process, an arriving update immediately starts service at one of $c$ homogeneous servers upon finding the system nonfull with exponentially distributed service times, and there is no waiting room. This system is known as the Erlang loss system in queueing theory \cite{erlang,book_kleinrock,book_grossAndharris} in which the number of busy servers increases by one upon an admitted arrival and decreases by one upon a service completion, making the system occupancy process a chain of birth-death type, and when arrivals finding all servers busy are blocked, the blocking probability is given by the Erlang-B formula \cite{book_kleinrock},\cite{book_grossAndharris}. We consider three update management policies which decide how to treat an arriving update when all the servers are busy. In the {\em non-preemptive} (NP) policy, the arriving update is blocked and not served at all, as in the original Erlang loss system. In the {\em preempt at random} (PR) policy, the arriving update starts receiving service by preempting one of the updates in service chosen uniformly at random, whereas the {\em preempt the stalest} (PS) policy preempts the particular server carrying the stalest update in the system. All three policies preserve the birth-death evolution of the occupancy process. Consequently, the probability of finding all servers busy is the same for these three policies and is given by the Erlang-B formula. Nevertheless, as we will show, the corresponding AoI and PAoI distributions of the three policies can be very different, as are their implementation complexities. The goal of this paper is the derivation of the {\em exact distributions} of AoI and PAoI under the three studied policies, for a single source as well as for multiple sources sharing the same pool of servers, in the latter case preemptions are restricted to updates of the same source, called self-preemption.

Since the development of the M/M/1 model in \cite{kaul_etal_infocom12} for the analysis of AoI in a single-source single-server scenario, most existing work makes use of queueing models to analyze AoI, where status updates are associated with update arrivals, and link or processing delays are tied to update service times \cite{RoyYates__AgeOfInfo_Survey}. Existing queueing models have targeted i) single-source or multi-source scenarios, ii) buffered or bufferless settings, iii) single-server and multi-server settings, and iv) single-hop and multi-hop scenarios, with various buffer management and/or scheduling mechanisms. Three methods stand out for the queueing analysis of AoI. The graph-based approach \cite{kaul_etal_infocom12} enables the calculation of the average AoI using graphical techniques in relatively simple systems. The systematic framework of stochastic hybrid systems (SHS) is proposed in \cite{yates_kaul_TIT19} to obtain the average AoI in more advanced systems where graphical methods fall short, and has been extended to yield the moments and the moment generating function of AoI in multi-source settings \cite{yates_mgf,moltafet_etal_tcom22}. The third method is the absorbing Markov chain (AMC) method that we build upon in this paper, first proposed in \cite{akar_gamgam_comlet23} to derive the distribution of AoI in a multi-source single-server scenario for both generate-at-will (GAW) and random arrival (RA) settings.

While the majority of existing work focuses on the average AoI, or on time averages of functions of the age \cite{tripathi_modiano_TNET}, a smaller body of work targets the exact {\em distribution} of the age, which is required whenever age violation probabilities, higher moments, or time averages of general functions of age are of interest. For single-server systems, \cite{inoue_etal_tit19} obtains a general formula for the stationary distribution of AoI and specializes it to a wide range of single-server queues and \cite{kesidis_etal_questa20} derives the distributions of AoI performance measures in bufferless update processing systems by means of the Palm inversion formula.  A follow-up Markov-renewal treatment of small-buffer systems in done in \cite{kesidis_etal_questa23} and related open problems are posed in \cite{kesidis_etal_questa22}. The stationary distribution of AoI in a discrete-time queue is obtained in \cite{kosta_etal_jsac21} whereas the distributions of age processes in networks of preemptive memoryless servers are studied in \cite{yates2020age}. In a separate work, the Laplace-Stieltjes transform of the AoI of a tagged stream sharing a single-server with Poisson background traffic is derived in \cite{inoue_mandjes_tit25}. In a recent work, single-server models in which updates enter the system with a nonzero initial age are investigated in \cite{miguelez_etal_arxiv25}. These distributional results, however, almost exclusively concern single-server systems. For general multi-server update systems, exact results have generally been limited to the average AoI, and the present paper aims to fill this gap.

A prominent mechanism to improve information freshness is to use server diversity, for which information sources use multiple-servers for submitting their updates. For related work on bufferless multi-server RA scenarios, the authors of \cite{kam_etal_TIT16} derive an expression for the average age for the single-source bufferless M/M/$\infty$ model, which can be viewed as the limiting behavior of the NP policy as $c \rightarrow \infty$. The reference \cite{yates_isit18} studies the average AoI for a single-source using a finite number of homogeneous servers by the SHS technique, and also the case of two heterogeneous servers; the update management scheme of \cite{yates_isit18} is identical to PS except that all obsolete updates are immediately discarded from their servers upon each service completion, a scheme which departs from the birth-death structure of the Erlang loss system studied here; we defer a detailed discussion of the consequences of this difference to Remark~\ref{rem:discard}. In \cite{javani_etal_globecom19}, the authors study a model where updates are directed towards the servers without taking their occupancies into account, and derive the average AoI using SHS for a single-source with two heterogeneous servers, as well as for a general number of sources with either two or three homogeneous servers; the extended version \cite{javani_etal_arxiv21} further provides a closed-form expression for the average AoI for a single-source and an arbitrary number of homogeneous servers under a per-server last-come-first-served policy with preemption, together with a recursive algorithm for the average AoI in the multiple-source multiple-server case. In contrast with \cite{javani_etal_globecom19,javani_etal_arxiv21}, in the present paper, the servers form a shared pool with pool-level admission and preemption decisions, and we obtain the exact distributions, rather than the averages, of both AoI and PAoI. On the scheduling side, \cite{bedewy_etal_tit19} establishes optimality properties of last-generated-first-served-type policies in multi-server update systems. 

In the domain of GAW multi-server systems, the authors of \cite{chen_etal_comlet23} use the SHS technique to obtain the average AoI for a heterogeneous dual-server system with exponentially distributed service times where obsolete packets stay in the system until they are received by the monitor. In \cite{chen_etal_twc23}, the authors extend their work to the case where one of the service times is deterministic, deriving closed-form expressions for the average AoI and average PAoI. The authors of \cite{akar_ulukus_tcom25} studies the distribution of both PAoI and AoI in heterogeneous dual-server GAW systems employing obsolete packet management either at the source by preemption, or at the monitor by discarding, when parallel/non-parallel transmission policies are used.

Our contributions in the current paper are summarized as follows:
\begin{itemize}
\item We present the AMC method in a self-contained manner aimed at the queueing community in the context of multi-server systems. The standard practice of modeling the sample-path of the underlying system by a recurrent Markov chain does not directly apply to the AoI process. Instead, the AMC method  models a {\em single AoI cycle} by an absorbing Markov chain and obtains the stationary laws of AoI and PAoI by Palm inversion arguments which is made precise in Lemma~\ref{lem:cycle}, upon which the proofs of our main theorems rest.
\item We derive the exact distributions of AoI and PAoI in single-source multi-server Erlang loss systems under the NP, PR, and PS update management policies. Both distributions are shown to be of matrix-exponential form.
The distributions are further shown to be of phase type through an explicit similarity transformation in Proposition~\ref{prop:phrep}. Obtaining the distributions enables the computation of time averages of arbitrary functions of the age, as well as age violation probabilities, in contrast to the majority of the existing multi-server literature which focuses on the average AoI/PAoI only.
\item Although the three policies are governed by the same birth-death occupancy process, and hence share the same Erlang-B blocking probability, their age distributions can differ substantially. This distinction is invisible to any occupancy-level analysis and is revealed only via age-level analysis proposed in this paper.
\item We extend the single-source analysis to $N$ sources sharing the pool of $c$ servers with self-preemption. A key structural property enabling this extension is that, under all three policies, a tagged source is affected by the remaining sources only through their aggregate Poisson update rate, so that the size of the state space is $O(c^3)$ regardless of the number of sources.
\end{itemize}

The paper is organized as follows. Section~\ref{sec:prel} presents the preliminaries on absorbing Markov chains (AMC) and matrix-exponential distributions. In Section~\ref{sec:systemmodel}, the system model and the update management policies are described. Section~\ref{sec:analysis} first describes the AMC method followed by the development of the analytical models yielding the distributions of AoI and PAoI for the three policies, the extension to the analysis to multiple-sources and  discussions on complexity, insensitivity, and extensions to phase-type service times. Comparative evaluation of the proposed policies is presented in Section~\ref{sec:numerical}. Conclusions, open problems, and future research directions are presented in Section~\ref{sec:conclusions}.

\section{Preliminaries}
\label{sec:prel}
\subsection{Notation}
We first present notation. Uppercase bold letters are used to denote real-valued matrices whereas uppercase letters denote random variables.
Lowercase bold (plain) letters or symbols, are used to denote real-valued vectors (scalars).
The $(i,j)^{\text{th}}$ element of a matrix $\bm{A}$ is denoted by $A_{i,j}$. The $j^{\text{th}}$ 
column of $\bm B$ is denoted by $\bm B(:,j)$ and the $i^{\text{th}}$ element of a row or column vector $\bm{a}$ is denoted by $a_i$. The spectral abscissa of the matrix $\bm A$ is the greatest real part among all of its eigenvalues, denoted by $\alpha(\bm A)$.
The notations $\bm{0}_{m \times n} $, ${\bm I_m}$, and ${\bm 1_m}$ denote a matrix of zeros of size $m \times n$, an identity matrix of size $m$, and a column vector of
ones of size $m$, respectively. When used without a subscript, size information is inferred from the context. 
\subsection{Continuous-time Absorbing Markov Chains}
Consider a continuous-time Markov chain (CTMC) $Z(t) \in \{1,2,\ldots,U+V \}, \ t \geq 0$, with $U \geq 1$ transient states and $V \geq 1$ absorbing states with an infinitesimal generator matrix $\bm{Q}$ with the following partitioning,
\begin{align}
\bm{Q} & = 
    \begin{bmatrix}
    \bm A & \bm B \\
    \bm 0 & \bm 0
    \end{bmatrix}, \label{generator}
\end{align}
where $\bm A$, called the transient matrix of the AMC, is a square matrix of size $U$ whose off-diagonal elements are composed of the transition rates among the transient states, whereas the diagonal entries of $\bm A$ are non-negative to make the row sums of $\bm Q$ zero. Moreover, $\bm A$ is a Hurwitz matrix, i.e., all the eigenvalues of $\bm A$ have strictly negative real parts. On the other hand, $\bm B$, called the exit-rate matrix, is a $U \times V$ matrix, comprises the transition rates from the transient states to the absorbing states.
The CTMC $Z(t)$ has the initial probability vector of size $1 \times U$ denoted by $\bm{\beta}$, 
\begin{align}
  \bm{\beta} & = \begin{bmatrix}
      \beta_1 & \cdots & \beta_U 
  \end{bmatrix},  \enspace \beta_i  = \mathbb{P} (Z(0) = i), \enspace  i=1,\ldots,U. \label{initialcondition}
\end{align}
The probability that we are at absorbing state $v$, $v=1,2,\ldots,V$, at time $t$ is given by,
\begin{align}
    \mathbb{P}(Z(t)=U+v) & = \bm{\beta} {({\mathrm e}^{\bm{A}t} - \bm{I})}\bm{A}^{-1}\bm{B}(:,v). \label{absorptionprobabilitiesattimet}
\end{align}
The probability that absorption occurs at absorbing state $v$, denoted by $p_v$, is given by \cite{kemeny1960finite},
\begin{align}
  p_v & = \mathbb{P}(Z(\infty)=U+v) = -\bm{\beta} \bm{A}^{-1} {\bm B}(:,v).  \label{absorptionprobabilitiesSS}
\end{align}
We write the transient probability vector of the AMC $Z(t)$ of size $1 \times U$ at time $t$ as follows,
\begin{align}
   \bm{z}(t) & = \begin{bmatrix} z_1(t) & \cdots & z_U(t) \end{bmatrix}, \enspace z_{i}(t)  = \mathbb{P} (Z(t)=i), \label{xt}
\end{align} 
for which the following closed-form expression holds,
\begin{align} \bm{z}(t) & = \bm{\beta} {\rm e}^{\bm{A}t}, \ t \geq 0. \label{ztclosed} \end{align} 
The absorption time $T$ of the CTMC $Z(t)$ has a cumulative distribution function (cdf) $F_T(t)= \mathbb{P} (T \leq t) $ and probability density function (pdf) $f_T(t)$, which can be written for $t \geq 0$,
\begin{align}
F_T(t) & = 1 - \bm{z}(t) \bm{1} = 1 - \bm{\beta} {\rm e}^{\bm{A}t} \bm{1}, \enspace f_T(t) = \dv{F_T(t)}{t} = \bm{\beta} {\rm e}^{\bm{A}t} \bm{b},  \label{cdf1pdf1} 
\end{align}
where $\bm b= \bm B \bm 1$ is the exit-rate vector. 
In this case, the absorption time $T$ is said to possess a continuous-time phase-type (CPH) distribution with representation $(\bm{\beta},\bm{A})$ which covers several well-known distributions including the exponential, hyper-exponential, Erlang, and Coxian distributions  
as its sub-cases \cite{neuts81}.
Finally, the $j^{\text{th}}$ moment of $T$ can be written in closed form through the following expression \cite{telek_book},
\begin{align}
    \mathbb{E}[T^j] & = j! \ \bm{\beta} (-\bm{A})^{-j-1} \bm{b},  \label{moments}
\end{align}
using the identity
\begin{align}
\int_0^{\infty} x^m \mathrm{e}^{{\bm A}x} \dd{x} = m! \, (-\bm{A})^{-m-1}, \ m=0,1,\ldots, \label{hurwitz}
\end{align}
for a Hurwitz matrix $\bm A$.

More generally, a non-negative random variable is said to possess a matrix-exponential (ME) distribution with representation $(\bm{v}, \bm{M}, \bm{w})$ if its pdf can be written in the form $f(x)= \bm{v}\, \mathrm{e}^{\bm{M} x} \bm{w}$ for $x \geq 0$, where the entries of the triple $(\bm{v}, \bm{M}, \bm{w})$ need not carry a probabilistic interpretation; the class of ME distributions strictly contains the class of CPH distributions \cite{nielsen_book}. For an ME distribution, all the moments follow from \eqref{hurwitz} when $\bm{M}$ is Hurwitz, and the Laplace-Stieltjes transform (LST) is available in the rational form
\begin{align}
\int_0^{\infty} \mathrm{e}^{-\theta x} \bm{v}\, \mathrm{e}^{\bm{M}x} \bm{w} \dd{x} & = \bm{v} \left( \theta \bm{I} - \bm{M} \right)^{-1} \bm{w}. \label{melst}
\end{align}
The AoI and PAoI distributions derived in this paper will all be of ME form with explicit representations, and will further be shown to be CPH. Consequently, moments, transforms, and exponential tail decay rates are all available in closed form, and the machinery of the matrix-analytic literature \cite{neuts81,latouche1999introduction,telek_book,nielsen_book} directly applies.

\section{System Model}\label{sec:systemmodel}
We first consider an update system involving a single-source and a monitor. The source sends updates (processed tasks, or status updates) towards a monitor by receiving service from one of the $c \geq 1$ homogeneous and independent servers. The updates are generated according to a Poisson process with rate $\lambda$ and forwarded to one of the free servers on a random basis where servers serve each of the updates with  random service time which is exponentially distributed with parameter $\mu$. Let us also define the offered load $a=\lambda/\mu$ and offered load per server $\rho=a/c$.
Each update contains the time stamp of the update, i.e., the time at which the update was generated by the source.

We will now describe the AoI process and PAoI processes for this M/M/c/c system. 
Let $d_l$ denote the instant at which the $l^{\text{th}}$ successful update is received by the monitor. As soon as a update is received, the monitor checks whether the time stamp of the newly received update is larger than the largest time stamp received so far, in which case the update is deemed successful (or up to date). The update is called obsolete if the packet is not successful. When an obsolete update is received, the monitor discards the update, and the AoI process is not subject to a jump at the reception instance. 
Let $t_l$ and $u_l=d_l - t_l$ denote the time stamp and the service time of the $l^{\text{th}}$ successful update, respectively.
Fig.~\ref{fig:samplepath} presents a sample path of the AoI process $\Delta(t)$ (thick solid curve). During cycle-$l$ with duration $\Psi_l$,  $\Delta(t)$ rises from the value $u_l$ at time $d_{l}$ until the value $\Phi_l$ at time $d_{l+1}$ when it drops to $u_{l+1}$. The random process $\Phi_l, l \geq 0$ is called the PAoI process since its value at cycle-$l$ is the peak value of the AoI process during the corresponding cycle. 
On the other hand, $\Psi_l=d_{l+1} - d_l, l \geq 0$, is called the AoI cycle length process. 
Since the AoI and PAoI processes are asymptotically stationary due to random interarrival and service times,
we let $\Delta$ and $\Phi$ denote the steady-state random variables for the random process $\Delta(t)$ and $\Phi_l$, respectively, with cdf 
\begin{align}
   F_{\Delta}(x) & = \lim_{t \rightarrow \infty} \mathbb{P} ( \Delta(t) \leq x), \enspace
   F_{\Phi}(x)  = \lim_{l \rightarrow \infty} \mathbb{P} ( \Phi_l \leq x)
\end{align}
for $x \geq 0$. We also denote by $f_{\Delta}$ and $f_{\Phi}$, the pdf of the AoI and PAoI processes, respectively. $\Psi$ denotes the steady-state random variable associated with the random process $\Psi_l$.

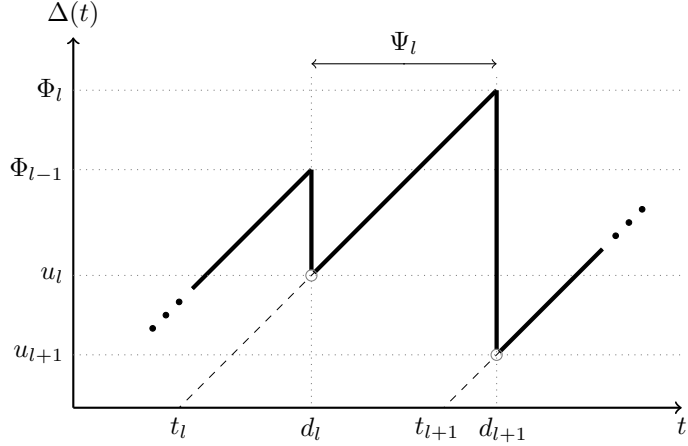
\begin{figure}[tb]
    \centering
    \begin{tikzpicture}[scale=0.35]	
    \draw[<->,black] (9,13) -- (16,13);
    \filldraw (12.5,13) circle (0.01) node[anchor=south, thick] {$\Psi_l$};
    \draw[thick,->] (0,0) -- (23,0) node[anchor=north] {$t$};
    \draw[thick,->] (0,0) -- (0,14) node[anchor=south] {$\Delta(t)$};
    \draw[ultra thick,black] (4.5,4.5) -- (9,9);
    \draw[dashed,very thin] (4.5,4.5) -- (9,9);
    \filldraw[black] (4,4) circle (3pt);
    \filldraw[black] (3.5,3.5) circle (3pt) ;
    \filldraw[black] (3,3) circle (3pt); 
    \draw (0,5) node[anchor=east] {$u_l$};
    \draw[dotted,gray] (0,9) -- (23,9);
    \draw[dotted,gray] (0,5) -- (23,5);
    \draw[dotted,gray] (0,2) -- (23,2);
    \draw[dotted,gray] (9,12.5) -- (9,0)  node[anchor=north, thick, black] {$d_{l}$};
    \draw[ultra thick,black] (9,9) -- (9,5.2);
    \draw[ultra thick,black] (9.1,5.1) -- (16,12);
    \draw[dashed,black,very thin] (16,12) -- (4,0)  node[anchor=north, thick, black] {$t_{l}$};
    \node at (3.75,0.50) (nodeA) {};
    \node at (8.75,5.5) (nodeB) {};
    \node at (15.75,12.5) (nodeC) {};
    \draw (0,12) node[anchor=east] {$\Phi_l$};
    \draw (0,9) node[anchor=east] {$\Phi_{l-1}$};
    \draw (0,2) node[anchor=east] {$u_{l+1}$};
    \draw[dotted,gray] (16,12.5) -- (16,0)  node[anchor=north, thick, black] {$\enspace d_{l+1}$};
    \draw[dotted,gray] (0,12) -- (23,12);
    \draw[ultra thick,black] (16,12) -- (16,2.2);
    \draw[ultra thick,black] (16.1,2.1) -- (20,6);
    \draw[dashed,black,very thin] (20,6) -- (14,0)  node[anchor=north, thick, black] {${t_{l+1}} \enspace$};
    \filldraw[black] (20.5,6.5) circle (3pt);
    \filldraw[black] (21,7) circle (3pt) ;
    \filldraw[black] (21.5,7.5) circle (3pt); 
    \draw[gray] (9,5) circle (6pt);
    \draw[gray] (16,2) circle (6pt);
    \end{tikzpicture}
    \caption{A portion of a sample path of the AoI process $\Delta(t)$, PAoI process $\Phi_l$, AoI cycle length process $\Psi_l$. Only generation and reception instances of successful packets are shown.}
    \label{fig:samplepath}
\end{figure}
An incoming update arrival will immediately receive service from one of the idle servers at random. 
In case the new update arrival occurs when all the $c$ servers are busy, we propose three update management policies. In the {\em non-preemptive} (NP) policy,
an incoming update is blocked, i.e., not admitted into the system, when all servers are busy. This system is the Erlang loss system of queueing theory \cite{erlang}, since updates are either admitted with no delay, or lost,  
differentiating itself from  buffered delay or delay-loss systems where incoming updates can also be exposed to certain delays \cite{book_kleinrock},\cite{book_grossAndharris}.
For finding the blocking probability, a CTMC $Y(t),t\geq 0$ is defined where $Y(t) \in \{ 0,1,\ldots,c \}$ stands for the number of busy servers at time $t$. It is clear for this CTMC that, the transition rate from state $i$ to $i+1$ is $\lambda$ for $0 \leq i < c$, and that from $i$ to $i-1$ is $i \mu$ for $1 \leq i \leq c$.  
Solving for the steady-state probabilities of the CTMC $Y(t)$ which is of birth-death type,
the steady-state probability $y_k$ that 
$k$ servers are busy for the Erlang loss system at an arbitrary time, was first given in closed form in \cite{erlang},
\begin{align}
    y_k & = \frac{\frac{a^k}{k!}}{\sum_{i=0}^c \frac{a^i}{i!}}, \ 0 \leq k \leq c, \label{ErlangLoss}
\end{align}
and the probability that an incoming update is blocked is $y_c$ since incoming updates are blocked when all the $c$ servers are busy. Moreover, the formula for $y_c$ in \eqref{ErlangLoss} holds for any service time distribution with mean service rate $\mu$, known as the insensitivity property of the blocking probability to the shape of the distribution \cite{book_grossAndharris}. 

The second policy is the {\em preempt at random} (PR) policy for which the incoming update preempts (or replaces) one of the updates in service, in a uniformly likely manner. With PR, since the time stamp of the new update is larger than all the time stamps in service, this new update is less likely to be obsolete than the update it replaces. Moreover, due to the memoryless nature of exponentially distributed service times, the remaining service time of the ongoing update and the service time of the new update are identically distributed, which reveals that there is actually no penalty paid with preemption. 

The final policy we study is the {\em preempt the stalest} (PS) for which the update to be preempted is the one which has the earliest time stamp. PS requires the update agent to be aware of the time stamps and to be able to order them in contrast to NP and PR, whereas PR requires the capability of preemption. This is in contrast to NP which is relatively easier to implement due to the lack of preemption and requirement for time stamp awareness. When service times are exponentially distributed, the probability that all the servers are busy for all the three policies is the same and given according to \eqref{ErlangLoss}, since the underlying CTMC $Y(t)$ that describes the evolution of the number of busy servers, is the same. The goal of this paper is the derivation of the exact distributions of AoI and PAoI for the Erlang loss system under the three update management policies, namely NP, PR, and PS.

\begin{remark}[Occupancy-preserving policies vs.\ immediate discarding]
\label{rem:discard}
Under each of NP, PR, and PS, an admitted update occupies exactly one server from its admission epoch until either its own service completion or its preemption by a newly admitted update. Consequently, the occupancy process moves in unit steps, and $Y(t)$ remains the birth-death chain of the Erlang loss system under all three policies, together with the attendant insensitivity of the blocking probability. This is in contrast with the update management scheme of \cite{yates_isit18}, in which every obsolete update is immediately discarded from its server the moment any update completes service in which case the occupancy may then decrease by more than one at a time, and the Erlang loss structure is lost. 
We nevertheless deliberately restrict our attention to the occupancy-preserving policies NP, PR, and PS, because they are precisely the policies for which the multi-source extension of Section~\ref{sec:multisource} is made possible for which a tagged source will then be affected by the remaining sources only through their aggregate Poisson update rate. Under immediate discarding, the number of servers occupied by each individual source would also need to be tracked, with the state space growing rapidly with the number of sources. The exact analysis of immediate discarding policies in the multi-source multi-server setting remains an open problem.
\end{remark}

In the multi-source version of the problem, there are $N \geq 1$ information sources submitting their updates to the same pool of $c$ homogeneous servers. Source-$n$ sends updates according to a Poisson process with intensity $\lambda_n$. The management scheme for an incoming update for source-$n$ is the same as the single-source case, but with self-preemption only, when all the servers are busy, i.e., the incoming update from source-$n$ can only preempt a source-$n$ update in service in the PR and PS policies upon finding the pool full.

\section{Analytical Model}
\label{sec:analysis}
\subsection{The AMC method}
\label{sec:amcmethod}
The standard practice in queueing theory for computing a stationary performance measure is to model the sample-path of the system by a recurrent (ergodic) Markov chain, and to obtain the measure of interest from the stationary distribution of this chain \cite{book_kleinrock,book_grossAndharris}. However, this practice cannot be used alone for the AoI process since $\Delta(t)$ is not a function of a finite system state, e.g., the occupancy $Y(t)$, since it also depends on the time stamps of the previously received updates. 
For AoI-related problems, one systematic approach is the SHS framework described in detail in \cite{yates_kaul_TIT19,yates_mgf} which has proven very effective for the computation of the mean, moments, and the moment generating function of the age, whereas the derivation of the distribution itself calls for different tools; see \cite{inoue_etal_tit19,kesidis_etal_questa20,kesidis_etal_questa23} for single-server systems.
On the other hand, the AMC method, first introduced in \cite{akar_gamgam_comlet23} and also adopted in this paper, takes a different route. In AMC, instead of the entire sample-path of the system, {\em a single AoI cycle} is modeled, by an {\em absorbing}, rather than a recurrent, Markov chain. Recall from Section~\ref{sec:systemmodel} and Fig.~\ref{fig:samplepath} that the sample-path of the AoI process decomposes into cycles delimited by the reception epochs $d_l$ of the successful updates. In particular, cycle-$l$ starts at $d_l$ with the age dropping to the value $u_l$, and ends at $d_{l+1}$ with the age attaining the peak value $\Phi_l$. The AMC method is based on the following three key ideas.

\paragraph{{\em Tagging:}} Every successful update opens exactly one AoI cycle. We therefore tag an arbitrary update $P^{\ast}$ admitted into service, and track the system from the admission epoch of $P^{\ast}$ onwards only through the finitely many quantities relevant to the fate of $P^{\ast}$: whether $P^{\ast}$ is still in service, and how many of the in-service updates carry time stamps earlier or later, than that of $P^{\ast}$. Owing to the Poisson arrivals and the memoryless service times, these quantities evolve as a finite-state continuous-time Markov chain.

\paragraph{{\em Two-stage AMC structure with two absorbing states:}} The tagged update $P^{\ast}$ may eventually be unsuccessful, i.e., preempted while in service (possible under PR and PS), or found obsolete at its own service completion (possible under all three policies), which we model by absorption into the state $u$.  Alternatively, $P^{\ast}$ may eventually be successful (service completes without becoming obsolete) in which case its reception by the monitor opens an AoI cycle, and we let the chain transition from the first stage (subset of transient states $P^{\ast}$ is present in the system) to the second stage (subset of transient states $P^{\ast}$ is not present any more in the system) which lasts until the reception of the first update fresher than $P^{\ast}$, necessarily the next successful update, closing the cycle, which we model by absorption into the state $s$. Specifically, we launch the clock of the AMC at the {\em generation} epoch of $P^{\ast}$. In this way, at the reception of $P^{\ast}$, the elapsed time of the AMC and the AoI both equal the service time of $P^{\ast}$, and both subsequently grow at unit rate until the cycle ends. Consequently, on successful realizations, the absorption time of the AMC equals the PAoI of the cycle opened by $P^{\ast}$, and the set of age values visited during that cycle coincides with the set of elapsed-time points at which the AMC resides in its second stage. 
In this manner, although the age process itself takes values in an unbounded continuous range, its behavior within a cycle is completely captured by which states of a finite absorbing chain are occupied and the computation of its stationary law reduces to elementary transient analysis, i.e., to matrix exponentials (MEs) of the transient matrix $\bm A$.

\paragraph{{\em PASTA property:}} Through the PASTA property \cite{book_grossAndharris}, the distribution of the system configuration found by $P^{\ast}$ at its admission epoch is the time-stationary distribution of the configuration, which yields the initial probability vector $\bm{\beta}$ of the AMC in closed-form in terms of the Erlang loss probabilities \eqref{ErlangLoss}.

These observations are made precise in Lemma~\ref{lem:cycle} below which, combined with the Palm inversion formula, yields the distributions of both PAoI (Theorem~\ref{thm:PAoI}) and AoI (Theorem~\ref{thm:AoI}) in ME form, once the generator of the AMC $X(t)$ is obtained as in \eqref{generator} together with its initial probability vector as in \eqref{initialcondition}. The construction is first developed in detail for NP; subsequently, the required modifications are elaborated for the other two policies PR and PS, and the extension to multiple sources is presented in Section~\ref{sec:multisource}.
\subsection{Non-preemptive Policy}
\begin{table}[t]
\caption{The state space for the AMC $X(t)$}
\centering
\begin{tabular}{|c|c|} 
 \hline
 State & Description \\
 \hline \hline
 $(i,j,1)$ & $0 \leq i,j \leq c-1, 0 \leq i + j < c$ \\
 \hline
$(i,j,2)$ & $0 \leq i \leq c-1, 0 \leq j  \leq c ,
0 \leq i + j \leq c$ \\
 \hline
 $u$ & Unsuccessful absorbing state  \\
 \hline
 $s$ & Successful absorbing state  \\ \hline 
\end{tabular}
\label{tab:states}
\end{table}
The state space of the problem we propose to use is given in Table~\ref{tab:states}.
The subset of transient states $\mathcal{L}_1 = \{ (i,j,1) \}, 0 \leq i,j \leq c-1, 0 \leq i + j < c$ constitute the first stage of the AMC when the specific update $P^{\ast}$ is in the system and is a legitimate candidate for successful reception, i.e., has not become obsolete yet, and there are $i$ (resp. $j$) updates in service with earlier (resp. later) time stamps than the time stamp of $P^{\ast}$. The condition $0 \leq i + j < c$ should hold since there can be at most $c$ busy servers together with $P^{\ast}$, 
when $P^{\ast}$ is in the system. Moreover, there are $L_1=c(c+1)/2$ states in the first stage.
On the other hand, for the subset of transient states $\mathcal{L}_2 = \{ (i,j,2) \}, 0 \leq i \leq c-1, 0 \leq j  \leq c,
0 \leq i + j \leq c$, constituting the second stage, $i$ and $j$ have the same interpretation as before, but $P^{\ast}$ has now successfully been received, and the AMC is to evolve until the next successful reception. Moreover, there are $L_2=\frac{(c+1)(c+2)}{2}-1$ states in the second stage.
The states $u$ and $s$ are the unsuccessful and successful absorbing states, respectively, indicating the eventual fate, successful or not, of the specific update $P^{\ast}$ alone, and not of any other update. 

The transition rates of the AMC $X(t)$ are given in Table~\ref{tab:transitionrates}. Let us first consider the transient state $(i,j,1)$. Under an arrival, $j$ will be incremented when the system is nonfull since its time stamp will be later than that of $P^{\ast}$. If one of the $i>0$ servers gets to complete, which occurs with rate $i \mu$, $i$ will be decremented. On the other hand, if one of the $j > 0$ servers completes its service, then $P^{\ast}$ becomes obsolete, i.e., it will eventually be discarded by the monitor, since an update with a later time stamp has been received by the monitor which makes $P^{\ast}$ obsolete. In this case, we let the AMC $X(t)$ absorb into the unsuccessful absorbing state $u$. When the service of the update $P^{\ast}$ is over, which occurs with rate $\mu$, we transition into $(i,j,2)$ since $P^{\ast}$ has successfully departed the system and the AMC needs to evolve until the next successful reception.
When $X(t)$ is in state $(i,j,2)$, $j$ will be incremented 
with a new update arrival which occurs with rate $\lambda$. If one of the $i>0$ servers completes its service, which occurs with rate $i \mu$, $i$ will be decremented. If one of the $j > 0$ servers completes its service, then the successful absorbing state $s$ is reached, since all the $j$ updates have been generated later than the generation time of the last successful packet, i.e., $P^{\ast}$. 
\begin{table}[t]
\caption{Transition rates for $X(t)$ for NP}
\centering
\begin{tabular}{|c|c|c|c|} 
 \hline 
 From & To  & Condition & Value  \\ 
 \hline
 $(i,j,1)$ & $(i,j+1,1)$ & $i + j < c-1$ & $\lambda$   \\ \cline{2-4}
               & $(i-1,j,1)$ & $i > 0$ & $i \mu$ \\ \cline{2-4} 
               & $u$ &  $j>0$ & $j \mu$ \\ \cline{2-4} 
               & $(i,j,2)$  & & $\mu$ \\ 
                 \hline
$(i,j,2)$ & $(i,j+1,2)$ &  $i + j < c$ & $\lambda$   \\ \cline{2-4} 
              & $(i-1,j,2)$ & $i > 0$ & $i \mu$ \\ \cline{2-4}
              & $s$ & & $j \mu$ \\
              \hline
\end{tabular}
\label{tab:transitionrates}
\end{table}

There are $L=L_1+L_2$ transient states and two absorbing states in the characterization of the AMC $X(t)$. We first deploy a suitable enumeration of the transient  states in Table~\ref{tab:states} while placing the absorbing states $u$ and $s$ at the end of the list of states, in this order. The enumeration process can be abstracted by a function $w(\cdot)$ which maps a joint state $(i,j,m)$ to an index $l, 1 \leq l \leq L$, i.e., $l=w(i,j,m)$, 
or equivalently, we use the notation $l=i|j|m$ instead.
In this setting, the states $u$ and $s$ are mapped to indexes $L+1$ and $L+2$, respectively. 
With a certain choice of the function $w(\cdot)$,  we obtain the generator matrix $\bm Q$ of the AMC $X(t)$ in the form \eqref{generator} where $\bm A$ is a square matrix of size $L$ and $\bm B$ is a $L \times 2$ matrix. Now the initial probability vector $\bm{\beta}$ of size $1 \times L$ needs to be obtained for the AMC $X(t)$. First, we observe that $P^{\ast}$ can start service only in the transient state 
$(i,0,1)$ for $0 \leq i \leq c-1$. Moreover, the probability that a new arrival can start service is $1-y_c$ where $y_c$ is the Erlang loss probability given in \eqref{ErlangLoss}. Consequently, the non-zero entries of the vector $\bm{\beta}$ can be written as,
\begin{align}
    \beta_{i|0|1} & = \frac{y_i}{1-y_c}, \ 0 \leq i \leq c-1,
\end{align}
where $\beta_{i|0|1}$ is the probability of finding $i$ busy servers conditioned on a nonfull system of $c$ servers and $y_i$ is given in \eqref{ErlangLoss}. At this point, we have completed the characterization of the AMC $X(t)$ through its transient matrix $\bm A$, exit-rate matrix $\bm B$, along with the initial probability vector $\bm{\beta}$. For convenience, we denote the second column of $\bm B$, the exit-rate vector corresponding to the successful absorbing state $s$, namely $\bm{B}(:,2)$, by $\bm s$. 
Next, the distributions of the AoI and PAoI processes will be given in closed-form using this characterization. The following lemma makes precise the reduction, outlined in Section~\ref{sec:amcmethod}, of the age process to the transient analysis of the AMC $X(t)$, and constitutes the backbone of Theorems~\ref{thm:PAoI} and \ref{thm:AoI}. Recall that the realizations of $X(t)$ being absorbed into $s$ are termed successful realizations, the remaining ones, absorbed into $u$, are termed unsuccessful realizations, and $T$ denotes the absorption time of $X(t)$.
\begin{lemma}
\label{lem:cycle}
Let $X(t)$ be the AMC characterized above with transient matrix $\bm A$, exit-rate matrix $\bm B$, initial probability vector $\bm{\beta}$,  $\bm s = \bm B(:,2)$ and absorption time $T$, launched at the generation epoch of an arbitrary tagged update $P^{\ast}$ admitted into service. Then,
\begin{enumerate}
\item[(i)] $X(t)$ is absorbed into $s$ if and only if $P^{\ast}$ is a successful update; moreover, the state $u$ is unreachable from the subset $\mathcal{L}_2$, so that $\{ X(x) \in \mathcal{L}_2 \} \subseteq \{ X(\infty) = s \}$ for every $x \geq 0$.
\item[(ii)] On successful realizations, if $P^{\ast}$ is the $l^{\text{th}}$ successful update, then $T = \Phi_l$, and for every $x \geq 0$, $X(x) \in \mathcal{L}_2$ holds if and only if $u_l \leq x < \Phi_l$, i.e., if and only if the age value $x$ is visited during cycle-$l$.
\item[(iii)] The steady-state PAoI satisfies $F_{\Phi}(x) = \mathbb{P} ( T \leq x \mid X(\infty) = s )$, and for every bounded measurable function $\varphi$, the steady-state AoI satisfies
\begin{align}
\mathbb{E} [\varphi(\Delta)] & = \frac{1}{\mathbb{E}[\Psi]} \int_0^{\infty} \varphi(x) \, \mathbb{P} ( X(x) \in \mathcal{L}_2 \mid X(\infty) = s ) \dd{x}, \label{palm}
\end{align}
where $\mathbb{E}[\Psi] = \int_0^{\infty} \mathbb{P} ( X(x) \in \mathcal{L}_2 \mid X(\infty) = s ) \dd{x}$ is the mean AoI cycle length.
\end{enumerate}
\end{lemma}
\begin{proof}
Items (i) and (ii) are pathwise consequences of the construction of $X(t)$. The transition from a state $(i,j,1)$ into $\mathcal{L}_2$ occurs exactly when the service of $P^{\ast}$ completes before $P^{\ast}$ has become obsolete, i.e., exactly when $P^{\ast}$ is received as a successful update, at which instant the AoI process experiences its downward jump. If $P^{\ast}$ is the $l^{\text{th}}$ successful update, this reception occurs at the epoch $d_l$, the AoI drops to the service time $u_l = d_l - t_l$ of $P^{\ast}$, and the elapsed time of the AMC, whose clock was launched at the generation epoch $t_l$ of $P^{\ast}$, equals the very same value $d_l - t_l = u_l$. From that instant onwards, both the AoI and the elapsed time of the AMC increase at unit rate. While $X(t)$ resides in $\mathcal{L}_2$, the indices $i$ and $j$ of the state $(i,j,2)$ count the in-service updates with time stamps earlier and later than $t_l$, respectively; a service completion among the $j$ later-stamped updates is necessarily a successful reception, since $P^{\ast}$ is, until that instant, the most recently generated update ever received by the monitor. This reception is precisely the epoch $d_{l+1}$ ending cycle-$l$, and it triggers absorption into $s$, whereas earlier-stamped completions and new arrivals keep $X(t)$ within $\mathcal{L}_2$. Consequently, $T = d_{l+1} - t_l = \Phi_l$, and $X(x) \in \mathcal{L}_2$ holds if and only if $u_l \leq x < \Phi_l$, proving (ii). Since no transition from $\mathcal{L}_2$ into $u$ exists (see Table~\ref{tab:transitionrates}; the same property is verified from Tables~\ref{tab:transitionratesPR} and \ref{tab:transitionratesPS} for the preemptive policies studied in the sequel), a realization visiting $\mathcal{L}_2$ at any time is necessarily a successful one, proving the second claim of (i). The first claim of (i) holds since absorption into $u$ occurs exactly when $P^{\ast}$ is preempted while in service or found obsolete at its own service completion, i.e., exactly when $P^{\ast}$ is unsuccessful.

For item (iii), consider the time-stationary and ergodic system process together with the stationary point process of admission epochs, and associate with the $n^{\text{th}}$ admitted update the system configuration found at its admission epoch. By the PASTA property, the Palm distribution of this configuration coincides with the time-stationary distribution of the configuration conditioned on admission, which is precisely the initial probability vector $\bm{\beta}$ constructed above; and by the memorylessness of the exponential service times and the Poisson arrivals, conditioned on this configuration, the subsequent evolution of the quantities relevant to the fate of the $n^{\text{th}}$ admitted update is that of the AMC $X(t)$, independently of the past. The ergodic theorem then implies that the long-run empirical distribution of the peaks $\{ \Phi_l \}$, i.e., of the absorption times collected over the successful updates, is the law of $T$ conditioned on $\{ X(\infty) = s \}$, proving the first claim of (iii). For the second claim, the Palm inversion formula (see, e.g., \cite{baccelli_bremaud}) applied to the stationary and ergodic point process $\{ d_l \}$ of successful receptions yields
\begin{align*}
\mathbb{E} [\varphi(\Delta)] & = \frac{\mathbb{E}^0 \left[ \int_{d_l}^{d_{l+1}} \varphi(\Delta(t)) \dd{t} \right]}{\mathbb{E}^0 \left[ \Psi_l \right]},
\end{align*}
where $\mathbb{E}^0$ denotes the Palm expectation with respect to $\{ d_l \}$; we also refer the reader to \cite{kesidis_etal_questa20}, where the Palm inversion formula is employed for the derivation of AoI distributions in bufferless single-server systems. Since $\Delta(t)$ increases with unit slope from $u_l$ to $\Phi_l$ within cycle-$l$, the integral in the numerator equals $\int_0^{\infty} \varphi(x) \chi_l(x) \dd{x}$, where $\chi_l(x)$ denotes the indicator of the event $\{ u_l \leq x < \Phi_l \}$, which, by item (ii), coincides with the indicator of the event $\{ X(x) \in \mathcal{L}_2 \}$ written for the AMC attached to the $l^{\text{th}}$ successful update. Averaging over the successful updates, i.e., taking the expectation with respect to the law of $X(t)$ with initial vector $\bm{\beta}$, conditioned on $\{ X(\infty) = s \}$, and interchanging the expectation and the integral by Tonelli's theorem, gives \eqref{palm}; the expression for $\mathbb{E}[\Psi]$ follows by choosing $\varphi \equiv 1$.
\end{proof}
First, let us present the distribution of the PAoI process for which Theorem~\ref{thm:PAoI} states our main result.
\begin{theorem}
\label{thm:PAoI}
Consider the AMC $X(t)$ 
characterized with the transient matrix $\bm A$, exit-rate matrix $\bm B$, the initial probability vector $\bm{\beta}$, and $\bm s=\bm{B}(:,2)$.
The steady-state PAoI $\Phi$ in the Erlang loss system with $c$ servers that deploys the NP update management policy possesses an ME distribution with pdf given by,
\begin{align}
f_{\Phi}(x)  & = \frac{\bm{\beta} \mathrm{e}^{\bm{A}x} \bm{s}}{-\bm{\beta} \bm{A}^{-1} \bm{s}},  \ x \geq 0, \label{density_paoi}
\end{align}
and the $j^{\text{th}}$ moment of $\Phi$ is written as,
\begin{align}
\mathbb{E} [\Phi^j] & = \frac{j! \ \bm{\beta} {(-\bm{A})^{-(j+1)}} \bm{s}}{-\bm{\beta} \bm{A}^{-1} \bm{s}}. 
\label{phi_twomoments}
\end{align}
\end{theorem}
\begin{proof}
By Lemma~\ref{lem:cycle}(iii), $F_{\Phi}(x) = \mathbb{P} ( T \leq x \mid X(\infty) = s )$ for $x \geq 0$. Since $s$ is an absorbing state, $\{ X(x) = s \} = \{ T \leq x, \, X(\infty) = s \}$, and therefore,
\begin{align}
F_{\Phi}(x) &  = \frac{\mathbb{P} ( X(x) = s) }{\mathbb{P} ( X(\infty) = s)} 
 = \frac{\bm{\beta} (\mathrm{e}^{\bm{A}x} - \bm{I}) \bm{A}^{-1} \bm{s}}{-\bm{\beta} \bm{A}^{-1} \bm{s}}, \label{phi3}
\end{align}
where the second equality stems from \eqref{absorptionprobabilitiesattimet} and \eqref{absorptionprobabilitiesSS}. Note that $\mathbb{P} ( X(\infty) = s) = -\bm{\beta} \bm{A}^{-1} \bm{s} > 0$, since the transition rate from any state $(i,j,1)$ into $\mathcal{L}_2$ is $\mu > 0$, and $s$ is reachable from every state of $\mathcal{L}_2$.
Differentiation of the expression \eqref{phi3} with respect to $x$ gives the expression 
\eqref{density_paoi}, and the expression for the moments of $\Phi$ in \eqref{phi_twomoments} follows from the identity \eqref{hurwitz}.
\end{proof}
Next, we present our main result on the distribution of the steady-state AoI and its moments in Theorem~\ref{thm:AoI}.
\begin{theorem}
\label{thm:AoI}
Consider the AMC $X(t)$ 
characterized with the transient matrix $\bm A$, exit-rate matrix $\bm B$, the initial probability vector $\bm{\beta}$, and the column vector $\bm{h}$ of size $L$, the indicator of the subset $\mathcal{L}_2$, whose non-zero entries are given by,
\begin{align}
{h}_{i|j|2} & = 1, \ 0 \leq i \leq c-1, \ 0 \leq j \leq c, \ i+j \leq c. \label{hdefine}
\end{align}
The steady-state AoI $\Delta$ in the Erlang loss system with $c$ servers that deploys the NP update management policy possesses an ME distribution with pdf given by, 
\begin{align}
f_{\Delta}(x)  & = \frac{\bm{\beta} \mathrm{e}^{\bm{A}x} \bm{h}}{-\bm{\beta} \bm{A}^{-1} \bm{h}},  \ x \geq 0, \label{density_aoi}
\end{align}
the $j^{\text{th}}$ moment of $\Delta$ is written as,
\begin{align}
\mathbb{E} [\Delta^j] & = \frac{j! \ \bm{\beta} {(-\bm{A})^{-j-1}} \bm{h}}{-\bm{\beta} \bm{A}^{-1} \bm{h}}, 
\label{Delta_twomoments}
\end{align}
and the mean AoI cycle length is given by,
\begin{align}
    {\mathbb{E}[\Psi]} & = \frac{\bm{\beta} {\bm{A}^{-1}} \bm{h}}{\bm{\beta} \bm{A}^{-1} \bm{s}}. \label{meancycle}
\end{align}
\end{theorem}
\begin{proof}
Fix $x \geq 0$. By Lemma~\ref{lem:cycle}(i), $\{ X(x) \in \mathcal{L}_2 \} \subseteq \{ X(\infty) = s \}$, and therefore,
\begin{align}
\mathbb{P} ( X(x) \in \mathcal{L}_2 \: | \: X(\infty)=s) & =  \frac{\mathbb{P} ( X(x)\in \mathcal{L}_2 ) }{\mathbb{P} ( X(\infty) = s)} = \frac{\bm{\beta} \mathrm{e}^{\bm{A}x} \bm{h}}{-\bm{\beta} \bm{A}^{-1} \bm{s}}, \label{occupancyL2}
\end{align}
where the second equality follows from the transient distribution \eqref{ztclosed} together with the definition \eqref{hdefine} of the indicator vector $\bm h$, and from \eqref{absorptionprobabilitiesSS}. Substituting \eqref{occupancyL2} into the identity \eqref{palm} of Lemma~\ref{lem:cycle}(iii) with arbitrary bounded measurable $\varphi$ shows that $\Delta$ possesses the pdf
\begin{align*}
f_{\Delta}(x) & = \frac{1}{\mathbb{E}[\Psi]} \ \frac{\bm{\beta} \mathrm{e}^{\bm{A}x} \bm{h}}{-\bm{\beta} \bm{A}^{-1} \bm{s}}, \ x \geq 0.
\end{align*}
Since $f_{\Delta}$ integrates to one, and $\int_0^{\infty} \bm{\beta} \mathrm{e}^{\bm{A}x} \bm{h} \dd{x} = -\bm{\beta} \bm{A}^{-1} \bm{h}$ by \eqref{hurwitz}, the normalization constant is identified as ${\mathbb{E}[\Psi]} \, (-\bm{\beta} \bm{A}^{-1} \bm{s}) = -\bm{\beta} \bm{A}^{-1} \bm{h}$, which yields \eqref{density_aoi} together with the expression \eqref{meancycle} for the mean AoI cycle length. The expression for the moments of $\Delta$ in \eqref{Delta_twomoments} follows from the identity \eqref{hurwitz}.
\end{proof}
The ME forms in Theorems~\ref{thm:PAoI} and \ref{thm:AoI} immediately yield the transforms and the tail behavior of the age processes.

\begin{proposition}[Phase-type representations]
\label{prop:phrep}
Let $\bm{\sigma} = -\bm{A}^{-1} \bm{s}$ and $\bm{\eta} = -\bm{A}^{-1} \bm{h}$, and let
$\bm{D}_{\sigma} = \mathrm{diag}(\bm{\sigma})$ and $\bm{D}_{\eta} = \mathrm{diag}(\bm{\eta})$.
Then, the steady-state PAoI $\Phi$ is CPH-distributed with representation $(\tilde{\bm{\beta}}, \tilde{\bm{A}})$, where
\begin{align}
\tilde{\bm{A}} & = \bm{D}_{\sigma}^{-1} \bm{A}\, \bm{D}_{\sigma}, \enspace
\tilde{\bm{s}} = \bm{D}_{\sigma}^{-1} \bm{s}, \enspace
\tilde{\bm{\beta}} = \frac{\bm{\beta} \bm{D}_{\sigma}}{\bm{\beta} \bm{\sigma}}, \label{phrepPAoI}
\end{align}
with $\tilde{\bm{s}}$ being the associated exit-rate vector. Similarly, the steady-state AoI $\Delta$ is CPH-distributed with representation $(\hat{\bm{\beta}}, \hat{\bm{A}})$, where
\begin{align}
\hat{\bm{A}} & = \bm{D}_{\eta}^{-1} \bm{A}\, \bm{D}_{\eta}, \enspace
\hat{\bm{h}} = \bm{D}_{\eta}^{-1} \bm{h}, \enspace
\hat{\bm{\beta}} = \frac{\bm{\beta} \bm{D}_{\eta}}{\bm{\beta} \bm{\eta}}, \label{phrepAoI}
\end{align}
with $\hat{\bm{h}}$ being the associated exit-rate vector.
\end{proposition}
\begin{proof}
The $i^{\text{th}}$ entry $\sigma_i$ of $\bm{\sigma}$ is the probability of successful absorption starting from the transient state $i$, which is positive for every $i$, since the transition rate from any state $(i,j,1)$ into $\mathcal{L}_2$ is $\mu>0$ and $s$ is reachable from every state of $\mathcal{L}_2$; hence $\bm{D}_{\sigma}$ is invertible, and the off-diagonal entries $a_{ij} \sigma_j / \sigma_i$ of $\tilde{\bm{A}}$ are non-negative, while $\tilde{\bm{\beta}}$ is a probability vector by construction. Writing the $i^{\text{th}}$ row of the identity $\bm{A} \bm{\sigma} = -\bm{s}$ as $a_{ii} \sigma_i + \sum_{j \neq i} a_{ij} \sigma_j = -s_i$, and dividing by $\sigma_i$, gives $a_{ii} + \sum_{j \neq i} \tilde{a}_{ij} + \tilde{s}_i = 0$, so that the pair $(\tilde{\bm{s}}, \tilde{\bm{A}})$ satisfies 
$\tilde{\bm{A}} \bm{1} +  \tilde{\bm{s}}=0$, $\tilde{\bm{A}}$ is a valid CPH transient matrix with exit-rate vector $\tilde{\bm{s}}$. Finally, since $\mathrm{e}^{\bm{D}_{\sigma}^{-1} \bm{A} \bm{D}_{\sigma} x} = \bm{D}_{\sigma}^{-1} \mathrm{e}^{\bm{A}x} \bm{D}_{\sigma}$, we have $\tilde{\bm{\beta}}\, \mathrm{e}^{\tilde{\bm{A}}x} \tilde{\bm{s}} = {\bm{\beta}\, \mathrm{e}^{\bm{A}x} \bm{s}}/({-\bm{\beta} \bm{A}^{-1} \bm{s}}) = f_{\Phi}(x)$ by Theorem~\ref{thm:PAoI}, proving the claim for $\Phi$. The claim for $\Delta$ follows verbatim with the pair $(\bm{\eta}, \bm{h})$ replacing $(\bm{\sigma}, \bm{s})$, upon noting that the $i^{\text{th}}$ entry $\eta_i$ of $\bm{\eta}$, the expected time spent in $\mathcal{L}_2$ starting from the transient state $i$, is positive for every $i$, and that $\bm{A} \bm{\eta} = -\bm{h}$.
\end{proof}

\begin{remark}[The $c=1$ sub-case]
\label{rem:subcase1} Let us visit the special case $c=1$ for which the state space of the problem is written as, 
\begin{align}
    \mathcal{L} & = \{(0,0,1),(0,0,2),(0,1,2),u,s \}, \label{cequals1}
\end{align}
With the ordering of the states as above, we express the characterizing matrices as,
\begin{align*}
    {\bm A} & = \begin{bmatrix}
        -\mu & \mu & 0 \\
        0 & -\lambda & \lambda \\
        0 & 0 & -\mu  
        \end{bmatrix}, \ {\bm s} = \begin{bmatrix}
            0 \\ 0 \\ \mu 
        \end{bmatrix},
        \ {\bm h} = \begin{bmatrix}
            0 \\ 1 \\ 1 
        \end{bmatrix},   
\end{align*}
and ${\bm{\beta}}  = \begin{bmatrix}
            1 & 0 & 0 
        \end{bmatrix}$. Using the two identities \eqref{phi_twomoments} and \eqref{Delta_twomoments}, we obtain
        \begin{align*}
      \mathbb{E}[\Phi] & = \frac{1}{\lambda} + \frac{2}{\mu}, \enspace   \mathbb{E}[\Delta]  = \frac{1}{\lambda} + \frac{2}{\mu}  - \frac{1}{\lambda + \mu},
        \end{align*}
        which overlap with the results obtained in \cite{costa_etal_TIT16} for the same system.
\end{remark}
        
\subsection{Preempt at Random Policy}
For PR, the state space of the problem is the same as in Table~\ref{tab:states}. However, we have additional transitions in comparison to NP, upon an update arrival when the system is full. These additional transition rates are presented in Table~\ref{tab:transitionratesPR}. Let us first consider the transient state $(i,j,1)$. Under an arrival upon a full system, $j$ will be incremented and $i$ will be decremented when $i>0$ with probability $i/(i+j+1)$. On the other hand, $P^{\ast}$ will be replaced with the new update with probability $1/(1+i+j)$, which then forces absorption into unsuccessful absorbing state $u$.
When $X(t)$ is in state $(i,j,2)$ and the system is full, i.e., $i+j=c$, for $i>0$, $j$ will be incremented and $i$ will be decremented with probability $i/(i+j)$. With these additional transition rates, the transient matrix $\bm A$ and exit-rate matrix $\bm B$ are characterized for PR. Let us now turn our attention to the initial probability vector $\bm{\beta}$. For this purpose, we observe that the incoming update $P^{\ast}$ is never blocked and it always will start receiving service. When there are $i, 0 \leq i < c-1$ busy servers, then $P^{\ast}$ will start service in the AMC state $(i,0,1)$. 
However, when there $i=c-1$ or $c$ busy servers, in each of these two cases, $P^{\ast}$ will start service in the AMC state $(c-1,0,1)$, since in the latter case, it will replace one of the ongoing updates. Mathematically, $\bm{\beta}$ is a row vector with its non-zero entries given as follows,
\begin{align}
\beta_{i|0|1} & = y_i, \enspace 0 \leq i < c-1, \enspace \beta_{c-1|0|1} = y_{c-1} + y_c. \label{betaPR}
\end{align}
Also, the column vector $\bm h$ is selected as in \eqref{hdefine}.
Now that the characterization of the AMC $X(t)$ is complete for PR, we observe that Lemma~\ref{lem:cycle}, and hence Theorems~\ref{thm:PAoI}, \ref{thm:AoI}  and Proposition~\ref{prop:phrep}, continue to hold verbatim: the entry of the AMC into $\mathcal{L}_2$ still coincides with the successful reception of $P^{\ast}$, the unsuccessful state $u$ remains unreachable from $\mathcal{L}_2$ (see Table~\ref{tab:transitionratesPR}), and the initial vector \eqref{betaPR} is again obtained from the PASTA property, now conditioning on the arrival being admitted with or without preemption. One can therefore employ \eqref{density_paoi} and \eqref{density_aoi} for writing the pdfs of the PAoI and AoI processes, respectively, for PR, from which any of the moments can be calculated.
\begin{table}[t]
\caption{Additional transition rates for $X(t)$ for PR when the system is full}
\centering
\begin{tabular}{|c|c|c|c|} 
 \hline 
 From & To  & Condition & Value  \\ 
 \hline
 $(i,j,1)$ & $(i-1,j+1,1)$ & $i + j = c-1, \enspace i>0$ & $ \frac{i}{i+j+1} \lambda $   \\ \cline{2-4}
           & $u$ & $i + j = c-1$ & $\frac{1}{i+j+1} \lambda$ \\
                 \hline
$(i,j,2)$ & $(i-1,j+1,2)$ &  $i + j = c,\enspace i>0$ & $\frac{i}{i+j} \lambda$   \\ 
              \hline
\end{tabular}
\label{tab:transitionratesPR}
\end{table}

\subsection{Preempt the Stalest Policy}
For PS, the state space of the problem is again the same as in Table~\ref{tab:states}. Similar to PR, we have additional transitions in comparison to NP, upon an update arrival when the system is full. These additional transition rates are presented in Table~\ref{tab:transitionratesPS}. Let us first consider the transient state $(i,j,1)$. Under an arrival upon a full system, $j$ will be incremented and $i$ will be decremented when $i>0$ since the stalest update will be replaced with the new arrival. On the other hand, when the system is full, $P^{\ast}$ will be replaced with the new update when $i=0$, which gives rise to absorption into unsuccessful absorbing state $u$.
When $X(t)$ is in state $(i,j,2)$ and the system is full, i.e., $i+j=c$, for $i>0$, $j$ will be incremented and $i$ will be decremented.  With the use of Table~\ref{tab:transitionratesPS}, we obtain the characterizing matrices $\bm A$ and $\bm B$ for the AMC $X(t)$ for PS.
Also, the column vector $\bm h$ is selected as in \eqref{hdefine}.
Moreover, the initial probability vector $\bm{\beta}$ is the same as in \eqref{betaPR}.
Having completed the characterization of the AMC $X(t)$ for PS, we note that Lemma~\ref{lem:cycle} again holds verbatim, by the same reasoning as for PR (see Table~\ref{tab:transitionratesPS}), and one can use \eqref{density_paoi} and \eqref{density_aoi} for writing the pdfs of the PAoI and AoI processes, respectively, for PS, from which any of the moments can be calculated.
\begin{table}[t]
\caption{Additional transition rates for $X(t)$ for PS when the system is full}
\centering
\begin{tabular}{|c|c|c|c|} 
 \hline 
 From & To  & Condition & Value  \\ 
 \hline
 $(i,j,1)$ & $(i-1,j+1,1)$ & $i + j = c-1, \enspace i>0$ & $ \lambda $   \\ \cline{2-4}
           & $u$ & $i + j = c-1, \enspace i=0$ & $ \lambda$ \\
                 \hline
$(i,j,2)$ & $(i-1,j+1,2)$ &  $i + j = c,\enspace i>0$ & $ \lambda$   \\ 
              \hline
\end{tabular}
\label{tab:transitionratesPS}
\end{table}

\begin{proposition}[Spectrum of the transient matrix]
\label{prop:spectrum}
For each of the policies NP, PR, and PS, order the transient states lexicographically by
(phase, $-i$, $j$). Then every transition of the AMC strictly increases this order, so
that $\bm{A}$ is permutation-similar to an upper triangular matrix; consequently, the
spectrum of $\bm{A}$ is real and equals the multiset of its diagonal entries. Every
diagonal entry is at most $\max(-\lambda, -c\mu)$, and the entries $-\lambda$ and
$-c\mu$ are attained at the states $(0,0,2)$ and $(0,c,2)$, respectively; hence
\begin{align}
\alpha(\bm{A}) & = -\min(\lambda, c\mu). \label{spectralabscissa}
\end{align}
\end{proposition}
\begin{proof}
Inspection of Tables~\ref{tab:transitionrates}, \ref{tab:transitionratesPR}, and
\ref{tab:transitionratesPS} shows that along every transition within the transient set,
the phase never decreases, the index $i$ never increases, and the index $j$ never
decreases within a phase. Arrivals increase $j$, completions of earlier-stamped updates
decrease $i$, the phase transition preserves $(i,j)$, and every admissible full-pool
preemption maps $(i,j)$ to $(i-1,j+1)$, while preemptions of later-stamped updates leave
the state unchanged. Hence, $\bm{A}$ is upper triangular under the stated order, and its
spectrum is the multiset of its diagonal entries, each being the negated total outflow
rate of the corresponding state: $-(\lambda + n\mu)$ for states with $n < c$ busy
servers, and, at full states, $-c\mu$ (NP), $-(\lambda(i+1)/c + c\mu)$ or
$-(\lambda i/c + c\mu)$ (PR), and $-(\lambda + c\mu)$ or $-c\mu$ (PS), all of which are
bounded above by $\max(-\lambda, -c\mu)$, with equality at the two stated states, proving
\eqref{spectralabscissa}. 
\end{proof}

\begin{remark}[The sub-case $c=1$]
\label{rem:subcase2}
For the sub-case $c=1$, PS is equivalent to PR. Similar to NP, we have the same state space as in \eqref{cequals1} and with the same ordering of the states, we have,
\begin{align*}
    {\bm A} & = \begin{bmatrix}
        -(\lambda + \mu) & \mu & 0 \\
        0 & -\lambda & \lambda \\
        0 & 0 & -\mu  
        \end{bmatrix}, \ {\bm s} = \begin{bmatrix}
            0 \\ 0 \\ \mu 
        \end{bmatrix},
        \enspace {\bm h} = \begin{bmatrix}
            0 \\ 1 \\ 1 
        \end{bmatrix}, 
\end{align*}
and ${\bm{\beta}}  = \begin{bmatrix}
            1 & 0 & 0 
        \end{bmatrix}$. Using the two identities \eqref{phi_twomoments} and \eqref{Delta_twomoments}, we obtain
        \begin{align*}
      \mathbb{E}[\Phi] & = \frac{1}{\lambda} + \frac{1}{\mu} + \frac{1}{\lambda + \mu},  \enspace \mathbb{E}[\Delta]  = \frac{1}{\lambda} + \frac{1}{\mu},
        \end{align*}
        which overlap with the results obtained in \cite{najm_nasser_isit16} for exponentially distributed service times.
        \end{remark}
\subsection{Extension to Multiple Sources}
\label{sec:multisource}
In this subsection, we will present the analytical model for the multi-source multi-server scenario. In this scenario, there are $N$ sources each submitting updates to the pool of $c$ servers according to a Poisson process with intensity $\lambda_n$ for source-$n$. 
The PR and PS policies only perform self-preemption, and cross-preemption is not allowed. Particularly, upon finding the pool full, the incoming update from source-$n$ preempts one of the servers serving source-$n$ updates either at random (PR), or the one serving the stalest source-$n$ update (PS). 
The steady-state AoI and PAoI of source-$n$ are denoted by $\Delta^{(n)}$ and $\Phi^{(n)}$, respectively.
Without any loss of generality, let us only consider the AoI/PAoI processes of source-$1$ since the sources can always be renumbered. 
For convenience, we let $\lambda = \lambda_1, \Delta=\Delta^{(1)}, \Phi=\Phi^{(1)}$.
We notice that the other sources affect the AoI process of source-$1$ only through their sum update rate $\lambda' = \sum_{n=2}^N \lambda_n$.  
We construct the AMC $X(t)$ for source-$1$ with the state space given in Table~\ref{tab:statesmultiplesources} where $k$ denotes the number of busy servers occupied with updates from all sources other than source-$1$. 
The number of states of the form $(i,j,k,1), 0 \leq i,j,k \leq c-1, i+j+k<c$ is equal to $L_1= c(c+1)(c+2)/6$, whereas the number of states of the form $(i,j,k,2), 0 \leq i \leq c-1, 0 \leq j,k \leq c, i+j+k \leq c$ is $L_2 = -1+(c+1)(c+2)(c+3)/6$.
The transition rates for the NP policy and the additional transition rates for the PR and PS policies are given in Table~\ref{tab:transitionratesmultiplesources}. 
 \begin{table}[t]
\caption{The state space for the AMC $X(t)$ for source-$1$ with update traffic from $N-1$ other sources}
\centering
\begin{tabular}{|c|c|} 
 \hline
 State & Description \\
 \hline \hline
 $(i,j,k,1)$ & $0 \leq i,j,k \leq c-1, \enspace 0 \leq i + j +k < c$ \\
 \hline
$(i,j,k,2)$ & $0 \leq i \leq c-1, \enspace 0 \leq j,k  \leq c, \enspace 0 \leq i + j + k \leq c$ \\
 \hline
 $u$ & Unsuccessful absorbing state  \\
 \hline
 $s$ & Successful absorbing state  \\ \hline 
\end{tabular}
\label{tab:statesmultiplesources}
\end{table}
\begin{table*}[t]
\caption{Transition rates for $X(t)$ in the multi-source scenario}
\centering
\begin{tabular}{|c|c|c|c|} 
 \hline 
 From & To  & Condition & Value    \\ 
 \hline
 \multicolumn{4}{|c|}{Transition rates for NP} \\ \hline
 $(i,j,k,1)$ & $(i,j+1,k,1)$ & $i + j +k < c-1$ & $\lambda$ \\ \cline{2-4}
             &  $(i,j,k+1,1)$ & $i + j +k < c-1$ & $\lambda'$  \\ \cline{2-4} 
               & $(i-1,j,k,1)$ & $i > 0$ & $i \mu$  \\ \cline{2-4} 
               & $u$ &  $j>0$ & $j \mu$  \\ \cline{2-4} 
               & $(i,j,k-1,1)$ &  $k>0$ & $k \mu$  \\ \cline{2-4} 
               & $(i,j,k,2)$  & & $\mu$  \\ 
                 \hline
$(i,j,k,2)$ & $(i,j+1,k,2)$ &  $i + j +k < c$ & $\lambda$   \\ \cline{2-4} 
            & $(i,j,k+1,2)$ &  $i + j +k < c$ & $\lambda'$   \\ \cline{2-4} 
              & $(i-1,j,k,2)$ & $i > 0$ & $i \mu$ \\ \cline{2-4}
               & $(i,j,k-1,2)$ & $k > 0$ & $k \mu$ \\ \cline{2-4}
              & $s$ & $j > 0$& $j \mu$ \\
              \hline
    \multicolumn{4}{|c|}{Additional transition rates for PR} \\ \hline           
       $(i,j,k,1)$ & $(i-1,j+1,k,1)$ & $i + j +k= c-1, \enspace i>0$ & $ \frac{i}{i+j+1} \lambda $   \\ \cline{2-4}
           & $u$ & $i + j+k = c-1$ & $\frac{1}{i+j+1} \lambda$ \\
                 \hline
$(i,j,k,2)$ & $(i-1,j+1,k,2)$ &  $i + j+k  = c, \enspace i>0$ & $\frac{i}{i+j} \lambda$   \\ 
              \hline  
      \multicolumn{4}{|c|}{Additional transition rates for PS} \\ \hline  
$(i,j,k,1)$ & $(i-1,j+1,k,1)$ & $i + j +k= c-1, \enspace i>0$ & $ \lambda $   \\ \cline{2-4}
           & $u$ & $i+j +k= c-1,\enspace i=0$ & $ \lambda$ \\
                 \hline
$(i,j,k,2)$ & $(i-1,j+1,k,2)$ &  $i + j +k= c, \enspace i>0$ & $ \lambda$   \\ 
              \hline    
\end{tabular}
\label{tab:transitionratesmultiplesources}
\end{table*}
Now the initial probability vector $\bm{\beta}$ needs to be obtained for the AMC $X(t)$. For this purpose, we define a two-dimensional recurrent Markov chain $W(t) \in \{  (u,v), 0 \leq u,v \leq c, 0 \leq u+v \leq c \}$ where $u$ and $v$ keep track of the number of busy servers with updates from source-$1$ and source-$1'$, respectively. The transition rate from state $(u,v)$ to $(u-1,v)$ is $ u \mu$ for $u>0$, and to $(u,v-1)$ is $v \mu$ for $v>0$, whereas the transition rate from $(u,v)$ to $(u+1,v)$ is $\lambda$ for $u>0$ and to $(u,v+1)$ is $\lambda'$ for $v>0$, when $u+v< c$. One can find the steady-state probabilities
\begin{align}
    \pi_{u,v} & = \lim_{t \rightarrow \infty} \mathbb{P} (W(t) = (u,v)), \\
     \pi_{u,v}' & = \lim_{t \rightarrow \infty} \mathbb{P} (W(t) = (u,v) | u+v \neq c), \\
     \pi_{u,v}'' & = \lim_{t \rightarrow \infty} \mathbb{P} (W(t) = (u,v) | v \neq c). \label{pidoubleprime}
\end{align}
As in the single-source case, we employ any enumeration of the $L=L_1 + L_2$ transient  states in Table~\ref{tab:statesmultiplesources} while placing the absorbing states $u$ and $s$ at the end of the list of states, in this order. The enumeration process can be abstracted by a function $g(\cdot)$ which maps a joint state $(i,j,k,m)$ to an index $l, 1 \leq l \leq L$, i.e., $l=g(i,j,k,m)$. 
However, we use the notation $l=i|j|k|m$ instead.
For the NP policy, the non-zero entries of the initial probability 
vector $\bm{\beta}$ can be written as,
\begin{align}
    \beta_{i|0|k|1} & = \pi_{i,k}', \enspace 0 \leq i,k \leq c-1, \enspace i+k < c,
\end{align}
where $\beta_{i|0|k|1}$ is the probability of finding $i$ busy servers with updates from source-$1$
and $k$ busy servers with updates from source-$1'$, conditioned on an nonfull system of $c$ servers.
On the other hand, for the PR and PS policies, the non-zero entries of their initial probability 
vector $\bm{\beta}$ are,
\begin{align}
\beta_{i|0|k|1} & = \pi_{i,k}'', \ i+ k < c-1, \\
\beta_{i|0|c-i-1|1} & =  \pi_{i,c-i-1}'' + \pi_{i+1,c-i-1}'', \ 0 \leq i \leq c-1.
\end{align}
To explain this, when the system is nonfull, a source-$1$ arrival joins one of the free servers, leaving the pair $(i,k)$ intact. When the system is full with at least one server carrying a source-$1$ update, the arrival is admitted by preempting one of those source-$1$ updates, so that a configuration $(u,v)=(i+1,c-i-1)$ at the arrival epoch leads to the initial state $i|0|c-i-1|1$. Finally, a source-$1$ arrival cannot join a server when all $c$ servers are occupied carrying source-$1'$ updates, which is the conditioning event excluded in \eqref{pidoubleprime}.
At this point, we have completed the characterization of the AMC $X(t)$ through its transient matrix $\bm A$, exit-rate matrix $\bm B$, the initial probability vector $\bm{\beta}$, and the column vector $\bm h$ for the three policies of interest. Lemma~\ref{lem:cycle} holds verbatim for the multi-source AMC as well. Successful receptions of source-$1$ updates delimit the AoI cycles of source-$1$, the entry into $\mathcal{L}_2$ coincides with the successful reception of the tagged source-$1$ update, the state $u$ is unreachable from $\mathcal{L}_2$, and the initial vector is obtained by the PASTA property applied to source-$1$ arrivals. Note that the aggregation of the other sources into a single Poisson stream of rate $\lambda'$ is exact here, since, under self-preemption, an arrival from source-$1'$ finding a full system either preempts another source-$1'$ update or is discarded, in both cases leaving the state $(i,j,k,m)$ of the tagged chain unchanged. Consequently, the expressions \eqref{density_paoi},\eqref{phi_twomoments},\eqref{density_aoi} and \eqref{Delta_twomoments}, as well as the the CPH representations of Proposition~\ref{prop:phrep}, can be used for the multi-source scenario, with the state space size $L = L_1 + L_2 = O(c^3)$ being independent of the number of sources $N$.

\subsection{Discussion}
\label{sec:discussion}
We close this section with a discussion of computational complexity and of two natural extensions.
For a single source, the number of transient states is $L = L_1 + L_2 = c^2 + 2c$, and for multiple sources, $L = O(c^3)$, independently of $N$. All the quantities appearing in Theorems~\ref{thm:PAoI} and \ref{thm:AoI}, namely $-\bm{\beta}\bm{A}^{-1}\bm{s}$, $-\bm{\beta}\bm{A}^{-1}\bm{h}$, the moments, and the transforms, require the solution of linear systems with the sparse matrix $\bm{A}$, whereas the evaluation of the pdfs at a given $x$ requires the evaluation of the matrix exponential $\mathrm{e}^{\bm{A}x}$. Both operations are routine for moderate state space sizes arising here even for large $c$.
\begin{remark}[On insensitivity]
\label{rem:insensitivity}
It is well known that the steady-state occupancy distribution \eqref{ErlangLoss} of the Erlang loss system, and in particular the blocking probability $y_c$, is insensitive to the service time distribution beyond its mean \cite{book_grossAndharris}. No such insensitivity can be expected for the AoI and PAoI distributions derived in this paper. Indeed, already for $c=1$, the mean AoI of the bufferless single-server system with preemption in service depends on the entire service time distribution and not only on its mean; see \cite{najm_nasser_isit16} where gamma-distributed service times are studied. Moreover, the AMC construction of this paper exploits the memorylessness of the service times in an essential way, e.g., through the exchangeability of remaining and fresh service times under preemption, and through the sufficiency of the counts $(i,j)$ for describing the fate of the tagged update. Quantifying the sensitivity of the age distributions in the Erlang loss system, whose occupancy distribution is insensitive, appears to be an open problem.
\end{remark}
\begin{remark}[Phase-type service times]
\label{rem:phservice}
The AMC method extends, in principle, to CPH-distributed service times with, say, $p$ phases: in addition to the counts $i$ and $j$ (and $k$ in the multi-source case), one needs to track how the earlier-stamped and later-stamped in-service updates are distributed over the $p$ service phases, i.e., the scalar counts are replaced by compositions, with the number of unordered configurations of $i$ updates over $p$ phases being $\binom{i+p-1}{p-1}$. The construction of the transition rates is then mechanical but the state space grows polynomially in $c$ with degree increasing in $p$, and the notation becomes cumbersome. Moreover, under PR and PS, preempting an update may no longer be advantageous for some service time distributions; see \cite{dogan_akar_tcom21} for some single-server results. We therefore do not pursue this extension here and leave it for future work.
\end{remark}

\section{Numerical Results}
\label{sec:numerical}
In this section, we present three numerical examples, all obtained directly from the exact expressions of Theorems~\ref{thm:PAoI} and \ref{thm:AoI}. No simulations are needed, although all the analytical results of this paper have additionally been cross-verified against independent discrete-event simulations. The examples illustrate, in turn, the distributional differences among the three policies, the use of the exact distributions for server provisioning under age violation constraints, and the use of the multi-source model for update rate allocation.

\begin{example}
\label{ex:ccdf}
We consider a single source with $\mu = 1$, and plot in Fig.~\ref{fig:example1} the complementary cdf $\mathbb{P}(\Delta > x)$ of the steady-state AoI on a logarithmic scale under the three policies, for the parameter pair $(\lambda, c) = (4,4)$ and $(8,8)$, both critically loaded with $\rho=1$, and the overloaded pair $(\lambda, c) = (16,8)$ with $\rho=2$. In all three plots, the PS curve lies below the PR curve, which in turn lies below the NP curve, for all values of $x$, suggesting an ordering of the three age distributions in the sense of first-order stochastic dominance for exponentially distributed service times. Two further observations are in order. First, the gains from preemption grow with the depth of the tail, and and also with the load. Second, the asymptotic behavior of the three curves is governed by
Proposition~\ref{prop:spectrum} as follows. Within each plot, all three policies share the same
exponential decay rate $\min(\lambda, c\mu)$, equal to $4$, $8$, and $8$, respectively,
so that the policies differ in the tail through the prefactor rather than through the
exponent. The prefactors need not be constants, since the dominant diagonal entry of
$\bm{A}$ need not be simple. We have numerically observed a purely exponential tail for
PR and PS whenever $\lambda \neq c\mu$, whereas the NP tail carries an additional
polynomial factor whenever $\lambda \geq c\mu$, which is apparent in the third plot where
the NP curve bends away from the other two. This is in line with the blocking mechanism
of NP, under which fresh updates cannot enter a full pool until the stale updates in
service drain one after another, lending the tail an Erlang-like character that
preemption removes.

Finally, we caution that the favorable ranking of the preemptive policies is tied to the exponentially distributed service times of the present paper. Particularly, for the single-server bufferless system, it is shown in \cite{dogan_akar_tcom21} that when the service time distribution approaches the deterministic distribution, preemption starts to become detrimental, and the non-preemptive policy can outperform its preemptive counterparts; in line with Remark~\ref{rem:insensitivity}, the policy rankings observed here should not be extrapolated to general service time distributions.
\end{example}

\begin{figure}[tb]
\centering
\includegraphics[width=0.9\linewidth]{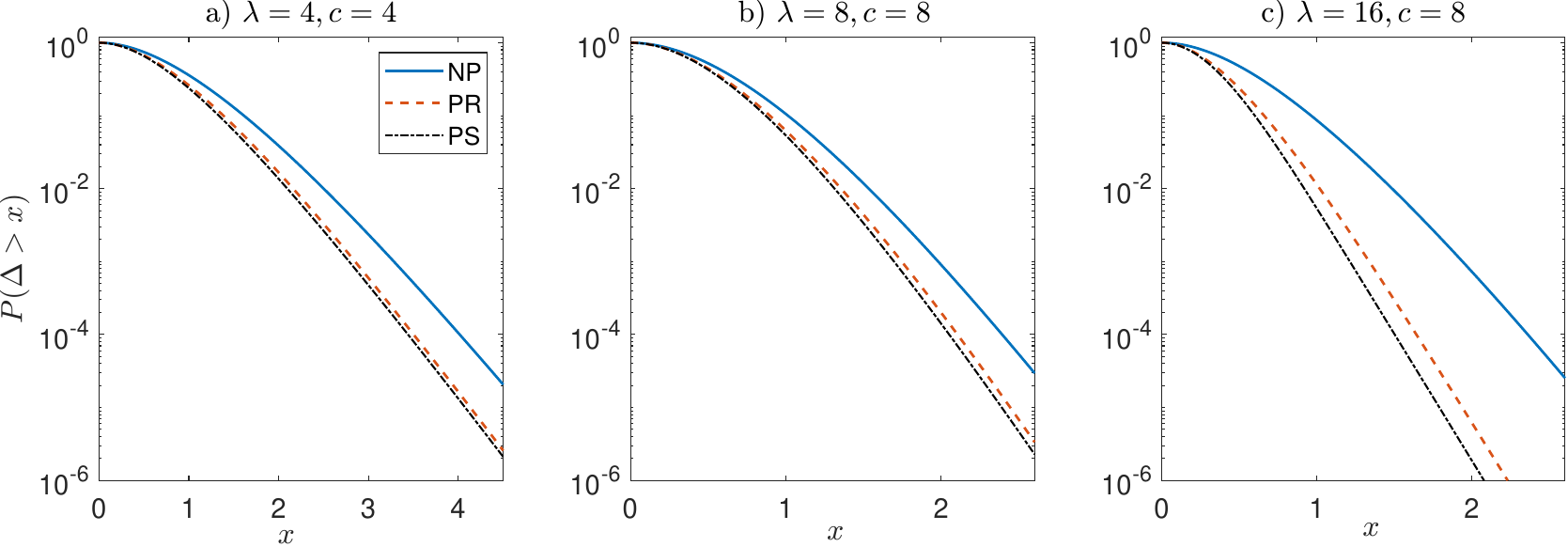}
\caption{Complementary cdf $\mathbb{P}(\Delta > x)$ of the steady-state AoI on a logarithmic scale under the NP, PR, and PS policies, for $\mu=1$, and $(\lambda,c)$ pair is chosen as a) $(4,4)$,  b) $(8,8)$, c) $(16,8)$.}
\label{fig:example1}
\end{figure}

\begin{example}
\label{ex:provisioning}
In this example, we consider the age violation constraint $\mathbb{P}(\Delta > \delta) \leq \epsilon$ as the performance metric of interest. For this purpose, Fig.~\ref{fig:example2}a and Fig.~\ref{fig:example2}b depict the quantity $\mathbb{P}(\Delta > \delta)$ as a function of the number of servers $c$, on a logarithmic scale, for $\mu = 1$, $\delta = 3$, and the update rates $\lambda = 6$ and $\lambda = 12$, respecively. Table~\ref{tab:provisioning} reports the smallest pool size $c^{\ast}$ meeting the constraint for violation tolerances $\epsilon = 10^{-1}, \ldots, 10^{-6}$. Three observations are in order. First, the preemptive policies meet every feasible target with one to three servers fewer than NP. Second, the returns on added servers diminish sharply beyond $c \approx \lambda/\mu$ (vertical dashed lines in Fig.~\ref{fig:example2}).  Third,  as $c \to \infty$, the pool is never full, so that all three policies coincide with the M/M/$\infty$ system studied in \cite{kam_etal_TIT16}. The M/M/$\infty$ floor equals $4.56 \times 10^{-6}$ for $\lambda = 6$ and $2.08 \times 10^{-11}$ for $\lambda = 12$ at $\delta = 3$ (horizontal dotted lines in Fig.~\ref{fig:example2}), towards which all three curves converge. Consequently, the constraint is achievable by server provisioning alone if and only if $\epsilon$ exceeds the floor for $\lambda = 6$, the target $\epsilon = 10^{-6}$ is infeasible for any number of servers, whereas for $\lambda = 12$, it is met with as few as six servers under PR and PS, as reported in Table~\ref{tab:provisioning}. For stricter targets, the update rate $\lambda$, rather than the size of the server pool, is the bottleneck.
\end{example}

\begin{figure}[tb]
\centering
\includegraphics[width=0.9\linewidth]{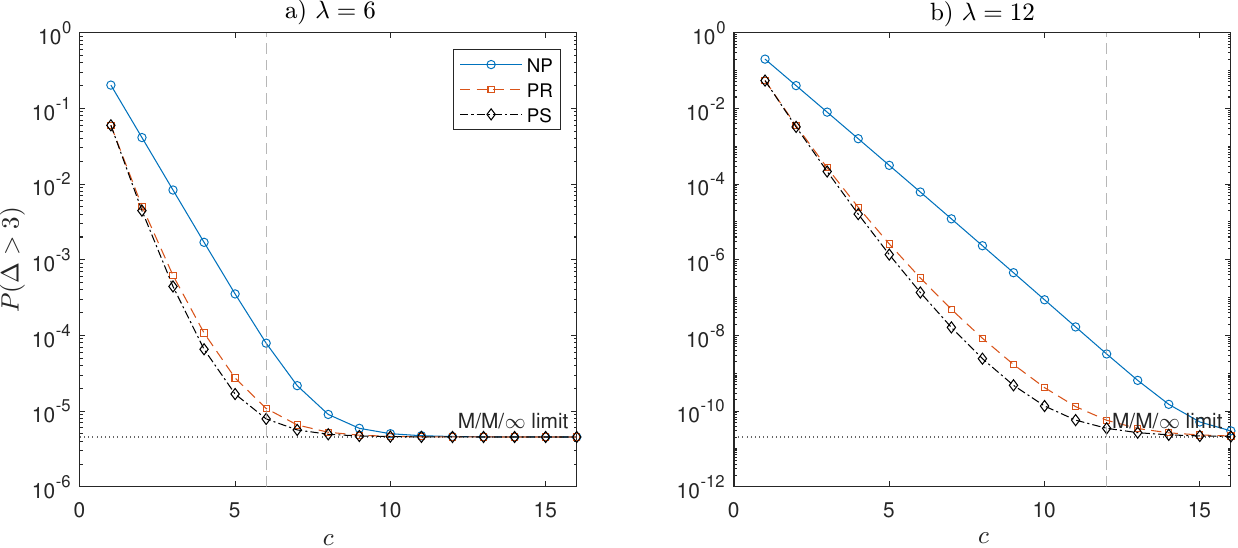}
\caption{Age violation probability $\mathbb{P}(\Delta > \delta)$ vs.\ the number of servers $c$ under the three policies, for $\mu=1$, $\delta = 3$, a) $\lambda = 6$, b) $\lambda = 12$. Horizontal dotted lines mark the M/M/$\infty$ floor and vertical dashed lines mark $c = \lambda/\mu$.}
\label{fig:example2}
\end{figure}

\begin{table}[tb]
\centering
\caption{Smallest number of servers $c^{\ast}$ such that $\mathbb{P}(\Delta > 3) \leq \epsilon$ for $\mu = 1$. The entry ``--'' indicates that the target is infeasible for any number of servers, since $\epsilon$ lies below the M/M/$\infty$ floor.}
\label{tab:provisioning}
\begin{tabular}{c|ccc|ccc}
\toprule
 & \multicolumn{3}{c|}{$\lambda = 6$} & \multicolumn{3}{c}{$\lambda = 12$} \\
$\epsilon$ & NP & PR & PS & NP & PR & PS \\
\midrule
$10^{-1}$ & 2 & 1 & 1 & 2 & 1 & 1 \\
$10^{-2}$ & 3 & 2 & 2 & 3 & 2 & 2 \\
$10^{-3}$ & 5 & 3 & 3 & 5 & 3 & 3 \\
$10^{-4}$ & 6 & 5 & 4 & 6 & 4 & 4 \\
$10^{-5}$ & 8 & 7 & 6 & 8 & 5 & 5 \\
$10^{-6}$ & -- & -- & -- & 9 & 6 & 6 \\
\end{tabular}
\end{table}

\begin{example}
\label{ex:allocation}
We now employ the multi-source model of Section~\ref{sec:multisource} for an update rate allocation problem. Consider $N=2$ sources sharing a pool of $c = 8$ servers with $\mu = 1$, and a total update rate budget $\Lambda$ to be split as $\lambda_1 + \lambda_2 = \Lambda$, $\lambda_n \geq 0$. For a weight parameter $\kappa \in [0,1]$, we seek the allocation minimizing the weighted mean AoI,
\begin{align}
\min_{\lambda_1 + \lambda_2 = \Lambda, \ \lambda_n \geq 0} \ \kappa \, \mathbb{E}[\Delta_1] + (1-\kappa) \, \mathbb{E}[\Delta_2], \label{weightedaoi}
\end{align}
where $\mathbb{E}[\Delta_n]$ denotes the mean AoI of source-$n$, obtained from Theorem~\ref{thm:AoI} with the multi-source AMC. By symmetry, it suffices to study $\kappa \in [1/2, 1)$. Fig.~\ref{fig:example3} depicts the optimal share $\lambda_1^{\ast}/\Lambda$ as a function of $\kappa$ under the three policies, for a) $\Lambda = 8$ b) $\Lambda = 16$. The split is even at $\kappa = 1/2$ and $\lambda_1^{\ast}$ increases with $\kappa$. However, the optimal allocation remains strictly interior for all $\kappa$, i.e., $\lambda_1^{\ast}/\Lambda < \kappa$.  The three policies prescribe visibly different allocations as the load grows: at $\Lambda = 16$ and $\kappa = 0.9$, the optimal shares are $0.788$, $0.764$, and $0.756$ under NP, PR, and PS, respectively, with NP tilting most aggressively towards the prioritized source in order to compensate for blocking. Moreover, PS attains a weighted mean AoI of $0.449$ at this point, about $34\%$ lower than the $0.678$ attained by NP. This example demonstrates that, once the exact multi-source analysis is available, optimization problems including update allocation become computationally feasible.  
For larger numbers of sources, for a fixed total budget, $\mathbb{E}[\Delta_n]$ depends on the rate vector only
through $\lambda_n$ itself. Therefore, the problem \eqref{weightedaoi} generalizes to a
\emph{separable} resource allocation problem  for which computationally efficient algorithms exist \cite{ibaraki_katoh_book,patriksson_ejor08}. A systematic study of such allocation algorithms in the scope of weighted AoI minimization is left for future research.
\end{example}

\begin{figure}[tb]
\centering
\includegraphics[width=0.9\linewidth]{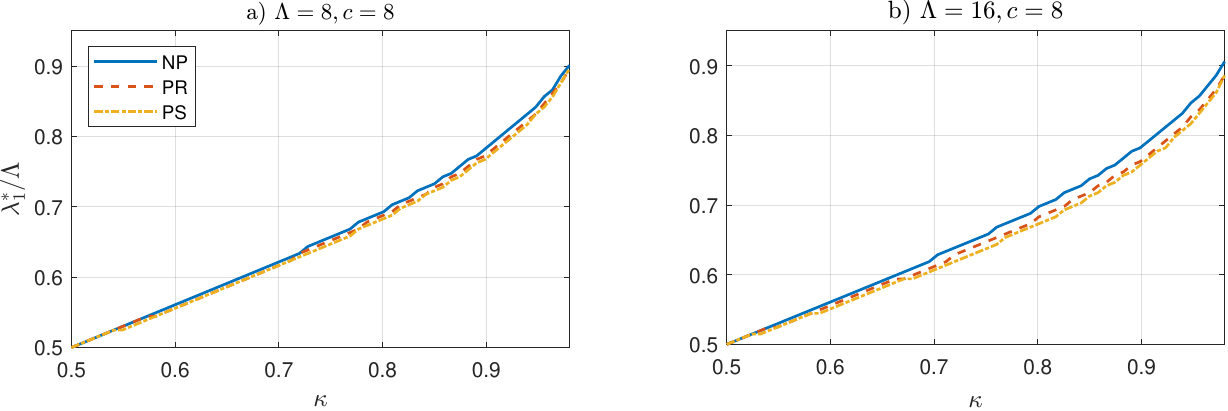}
\caption{Optimal rate share $\lambda_1^{\ast}/\Lambda$ minimizing the weighted mean AoI \eqref{weightedaoi} vs.\ the weight $\kappa$, for two sources sharing $c=8$ servers with $\mu = 1$; left: $\Lambda = 8$, right: $\Lambda = 16$.}
\label{fig:example3}
\end{figure}

\section{Conclusions}
\label{sec:conclusions}

The distributions of AoI and PAoI in a single-source multi-server Erlang loss system are derived using the absorbing Markov chain method for three update management policies which are non-preemptive (NP), preempt at random (PR), and preempt the stalest (PS) policies. For exponentially distributed service times, for a given desired age violation probability, we show that PS requires the fewest number of servers whereas the simple-to-implement NP is the least performing policy. We have also extended the method to model the case of multiple sources under sel-preemption, where the AoI/PAoI processes of a given source are shown to be impacted by the overall update rate of exogenous update traffic from other sources.

Several directions for future research emerge from this work. First, the exact distributions enable optimization formulations beyond the mean, e.g., server provisioning or per-source rate allocation under age violation constraints. Second, the exact analysis of immediate-discarding policies in the multi-source multi-server setting remains open (Remark~\ref{rem:discard}), as does the quantification of the sensitivity of the age distributions to the service time distribution (Remark~\ref{rem:insensitivity}) and the extension to phase-type service times (Remark~\ref{rem:phservice}). Third, allowing cross-source preemption, prioritizing certain sources at the expense of others, calls for a different state description and is left for future work. Finally, our numerical results suggest orderings among the age distributions of NP, PR, and PS; establishing such comparisons in the sense of stochastic orders would be of independent interest.

\bibliographystyle{plain}
\bibliography{bibl}
\end{document}